\documentclass[twoside,11pt]{article}

\usepackage{blindtext}

\usepackage[abbrvbib, preprint]{jmlr2e}

\usepackage{amsthm}
\usepackage{physics}
\usepackage{bbm}

\usepackage[nameinlink, noabbrev]{cleveref}
\newtheorem{theorem}{Theorem}

\theoremstyle{definition}
\newtheorem{example}{Example}%

\Crefname{corollary}{Corollary}{Corollaries}
\Crefname{lemma}{Lemma}{Lemmas}
\Crefname{figure}{Figure}{Figures}
\Crefname{definition}{Definition}{Definitions}
\Crefname{inequality}{inequality}{inequalities}
\Crefname{example}{Example}{Examples}
\Crefname{proposition}{Proposition}{Propositions}
\Crefname{theorem}{Theorem}{Theorems}

\usepackage{subcaption}

\DeclareMathOperator*{\E}{E}
\DeclareMathOperator*{\argmin}{\arg\,\min}

\newcommand{\given}{\mid}
\newcommand{\altgiven}{;\,}
\newcommand{\thetahat}{\widehat{\theta}}
\newcommand{\rhat}{\widehat{R}}

\newcommand{\Gbar}{\overline{G}}

\usepackage{lastpage}

\ShortHeadings{Multiple Testing in Generalized Universal Inference}{Dey, Martin, and Williams}
\firstpageno{1}

\begin{document}

\title{Multiple Testing in Generalized Universal Inference}

\author{\name Neil Dey \email ndey3@ncsu.edu\\
\name Ryan Martin \email rgmarti3@ncsu.edu \\
\name Jonathan P. Williams \email jwilli27@ncsu.edu \\
       \addr Department of Statistics\\
       North Carolina State University\\
       Raleigh, NC 27607-6698, USA
}

\maketitle

\begin{abstract}
Compared to p-values, e-values provably guarantee safe, valid inference.  If the goal is to test multiple hypotheses simultaneously, one can construct e-values for each individual test and then use the recently developed e-BH procedure to properly correct for multiplicity.  
Standard e-value constructions, however, require distributional assumptions that may not be justifiable. This paper demonstrates that the generalized universal inference framework can be used along with the e-BH procedure to control frequentist error rates in multiple testing when the quantities of interest are minimizers of risk functions, thereby avoiding the need for distributional assumptions. We demonstrate the validity and power of this approach via a simulation study, testing the significance of a predictor in quantile regression.
\end{abstract}


\section{Introduction}\label{sec:intro}
In scientific studies, it is common to have several hypotheses to be tested simultaneously. This is particularly common in exploratory studies where scientists aim to discover which of many measured covariates impact a response, or in situations where an entire set of inferences must be simultaneously correct to guarantee a correct ``overall'' decision. In such cases, even if many of the null hypotheses are true, the number of null hypotheses that will be rejected due to chance will increase with the number of tests, at any fixed level $\alpha$.  Hence, methods such as the Benjamini--Hochberg procedure \citep{benjamini1995} and the Benjamini--Yekutieli procedure \citep{benjamini2001} have been developed to control the \textit{false discovery rate}---the expected proportion of falsely rejected null hypotheses---when
testing for statistical significance via p-values. 

In recent years, there has been a growing interest in using \textit{e-values} for hypothesis testing. A random variable $W$ is an e-value if $\E(W) \leq 1$, where the expectation is taken with respect to the null hypothesis distribution.  Since the reciprocal of an e-value is a p-value, it can be readily used for uncertainty quantification similarly to p-values. However, e-values possess a variety of advantages over p-values, e.g., they can be more easily combined than p-values, they remain valid under the practice of ``optional continuation," and they enjoy certain robustness properties. For more complete lists of {\em e-over-p} advantages, see \citet[][Sec.~2]{wang2022} and \citet[][Sec.~1.5]{ramdas2024}; see also: \citet{shafer2011}, \citet{chung2023}, \citet{vovk2023}, and \citet{grunwald2024}. The ``e-BH" procedure of \citet{wang2022} modifies the Benjamini--Hochberg procedure---which was designed for a set of \textit{independent} p-values---to be applicable to a general set of \textit{arbitrarily dependent} e-values, controlling the false discovery rate when
testing multiple hypotheses. 

The e-values used in multiple testing applications commonly assume a correctly specified statistical model \citep[e.g.,][]{ignatiadis2023, nguyen2023}. There are a variety of statistics and machine learning applications in which such an approach is not practical: For example, practitioners may be loath to make such distributional assumptions, or perhaps their inferential target is not
the parameter of a data-generating model at all, as in quantile regression \citep{koenker1978} or inference on the minimum clinically important difference \citep{hedayat2015,syring2017}. In such cases, it natural to consider the quantities of inferential interest 
as minimizers of some more general \textit{risk} function. While
uncertainty quantification for risk minimizers has been extensively studied (e.g., via M-estimation \citep{huber1981} or more recent approaches such as \citealt{hudson2021} and \citealt{cella2022}), these approaches do not generally offer finite-sample validity guarantees. 

In this paper, we pair \textit{generalized universal inference} \citep{dey2024} with the e-BH procedure to create e-values that possess finite-sample validity guarantees in a multiple testing problem when the target of inference is the minimizer of a risk function. We then apply this procedure to test for the significance of a predictor in quantile regression, and we investigate our proposed test's power via a simulation study---the code for which is available at \url{https://github.com/neil-dey/gue-multiple-testing}.

\section{Generalized Universal Inference}
An important recent development in statistics is \textit{universal inference} \citep{wasserman2020}, which constructs e-values that can be used for safe or finite-sample valid inference under no regularity conditions. In their proposed framework, one splits the sample $S$ into a training set $S_1$ and a validation set $S_2$ and computes a maximum likelihood estimator $\thetahat_{S_1}$ on $S_1$; then, supposing that the likelihood function $L$ is correctly specified, a $1-\alpha$ level confidence set for the data-generating parameter $\theta^*$ is given by $\{\theta : T(\theta) \leq \alpha^{-1}\}$, where
\begin{equation*}
    T(\theta) := \frac{L(\thetahat_{S_1}\altgiven S_2)}{L(\theta\altgiven S_2)}.
\end{equation*}
Indeed, the universal inference test statistic $T$ is an example of an e-value, thus providing safe inference under no assumptions other than a correctly specified likelihood function.

In certain statistics and machine learning applications, the quantity of interest is not the parameter of a statistical model.  In quantile regression, for example, the true regression coefficients are define to minimize the expected check loss.  Positing a statistical model to solve such a problem creates a risk of model misspecification bias, introduces nuisance parameters, etc., none of which are advantageous.  In such cases, \textit{generalized universal inference} \citep{dey2024} can be used instead to directly construct e-values and achieve finite-sample validity.  In this approach, one measures how well a value $\theta\in\Theta$ of the quantity of interest conforms with an observed datum $z\in\mathcal{Z}$ using a prespecified \textit{loss function} $\ell:\Theta\times\mathcal{Z} \rightarrow \mathbb{R}$, with smaller values of $\ell(\theta\altgiven z)$ indicating a greater degree of conformity between $\theta$ and $z$. Next, define the \textit{risk function} to be the expected value of the loss, $R(\theta) := \E[\ell(\theta\altgiven Z)]$, so that the target of inference is the risk minimizer:
\begin{equation*}
    \theta^* := \arg\min_{\theta\in\Theta}R(\theta).
\end{equation*}
A point estimate for $\theta^*$ given data $Z_1, \ldots, Z_n$ is given by the \textit{empirical risk minimizer} (ERM), $\thetahat_n := \arg\min_{\theta\in\Theta} \rhat_n(\theta)$, 
where $\rhat_n(\theta) := n^{-1}\sum_{i=1}^n \ell(\theta \altgiven Z_i)$ is the \textit{empirical risk}.  The \textit{generalized universal e-value} (GUe-value, pronounced ``gooey value'') is then defined to be
\begin{equation*}
    G_n(\theta) := \exp\qty[-\omega \cdot |S_2|\cdot  \qty(\rhat_{S_2}(\thetahat_1) - \rhat_{S_2}(\theta))],
\end{equation*}
where $S$ (a sample of size $n$) is partitioned into $S_1 \sqcup S_2$, $\thetahat_{S_1}$ is the ERM on $S_1$, $\rhat_{S_2}$ is the empirical risk function on $S_2$, and $\omega > 0$ is a \textit{learning rate} hyperparameter. As shown in \citet{dey2024}, the GUe-value is a valid e-value for sufficiently small $\omega$ if the \textit{strong central condition} holds, i.e., if there exists $\bar{\omega} > 0$ such that
\begin{equation*}
    \E \exp\qty[-\omega\{\ell(\theta\altgiven Z) - \ell(\theta^*\altgiven Z)\}] \le 1 \quad \text{for all } \theta\in\Theta \text{ and } \omega \in [0, \bar{\omega}).
\end{equation*}
In this case, Theorem~1 of \citet{dey2024} implies that $\{\theta : G_n(\theta) < \alpha^{-1}\}$ is a $1-\alpha$ level confidence set for $\theta^*$, and that the test that rejects $H_0: \theta^* \in\Theta_0$ 
if and only if $
G_n(\Theta_0) := \inf_{\theta\in\Theta_0} G_n(\theta) \geq \alpha^{-1}$, controls the Type I error at level $\alpha$. 

Despite its name, the strong central condition is not particularly restrictive in most applications; see Section 3 of \citet{dey2024} for further discussion, as well as \citet{vanErvan2015}, \citet{heide2020}, and \citet{grunwald2020}.  It is true, however, that the validity and efficiency of the GUe-value hinges on the choice of the learning rate, so it is suggested in \citet[][Sec.~4]{dey2024} to apply learning rate selection algorithms developed for Gibbs posteriors to choose learning rates for GUe-values (e.g., those proposed in \citealt{bissiri2016}, \citealt{holmes2017}, \citealt{grunwald2017}, and \citealt{lyddon2019}).  The authors specifically suggest use of a bootstrap-based algorithm inspired by the ``General Posterior Calibration'' algorithm \citep{syring2019, martin2022} to select the necessary learning rate, as this algorithm empirically maintains frequentist validity at finite sample sizes.  Finally, if the loss function is precisely the negative log-likelihood function for the data, then generalized universal inference is universal inference, and the strong central condition holds with $\bar{\omega} = 1$.

\section{Generalized Universal Multiple Testing}
There are two situations of interest: First, the situation in which several GUe-values are of standalone interest, and we require control of the false discovery rate (FDR); second, the situation in which several GUe-values are combined into a single e-value to determine if \textit{any} of the tested null hypotheses are false. In this section, we demonstrate that the e-BH procedure of \citet{wang2022} can be used to accomplish both of these goals using GUe-values.

In the e-BH procedure, we are given e-values $e_1, \ldots e_M$ corresponding to tests of $M$ many null hypotheses, $H_0^m: \theta_m \in \Theta_0^m$ for $m=1,\ldots,M$. Assuming these e-values have been sorted from largest to smallest, we consider the transformed e-values ${e^*_m := m\cdot e_m/M}$.
\citet{wang2022} prove that these new e-values $e^*_1, \ldots, e^*_M$ control the FDR under arbitrary dependence between the original e-values. Consequently, so long as the strong central condition holds, one can simply apply this e-BH procedure to a set of GUe-values to control the FDR at a prespecified level. 

For our second goal of combining GUe-values into a single hypothesis test to determine if any of the original null hypotheses are false, we may merge the e-BH, FDR-controlled e-values using any of the strategies explored in \citet{vovk2021}. We find the most useful merging function is simply the average of the e-values because this can be shown to ``essentially dominate'' any other symmetric merging function. We thus see that given the (sorted) GUe-values $\{G_{n_m}^{(m)}(\Theta_0^m): m=1,\ldots,M\}$, the final GUe-value
we propose for this multiple testing situation is
\begin{equation}\label{eq:finalGue}
\Gbar_M(\Theta_0^1,\ldots,\Theta_0^M) := \frac{1}{M}\sum_{m=1}^M \frac{m}{M}\cdot G_{n_m}^{(m)}(\Theta_0^m).
\end{equation}

\begin{theorem}\label{thm:size}
Let $\{G_{n_m}^{(m)}(\Theta_0^{(m)}): m=1,\ldots,M\}$ be sorted GUe-values, where each GUe-value uses a learning rate that satisfies the strong central condition. If all of $H_0^{(1)},\ldots,H_0^{(M)}$ are true, then $\Gbar_M$ in \eqref{eq:finalGue} satisfies 
\begin{equation*}
    \Pr( \Gbar_M(\Theta_0^1,\ldots,\Theta_0^M) \geq \alpha^{-1}) 
    \leq \alpha.
\end{equation*}
\end{theorem}
\begin{proof}
Virtually identical to the proof of Theorem 2 of \citet{wang2022}.
\end{proof}

\begin{theorem}\label{thm:power}
Let $\{G_{n}^{(m)}(\Theta_0^{(m)}): m=1,\ldots,M\}$ be the sorted GUe-values where, for simplicity, each test is based on a sample of size $n$.  Suppose that $\sup_{\vartheta\in\Theta^m}|\rhat^{(m)}_n(\vartheta) - R^{(m)}(\vartheta)| = o_p(1)$ as $n \to \infty$ for each $m=1,\ldots,M$.  If there exists $m$ such that $H_0^{(m)}$ is false, i.e., $\inf_{\vartheta \in \Theta_0^m} R^{(m)}(\vartheta) > R^{(m)}(\theta_m^*)$, 
then $\Gbar_M$ in \eqref{eq:finalGue} satisfies 
\begin{equation*}
    \lim_{n\rightarrow\infty}\Pr( \Gbar_M(\Theta_0^1,\ldots,\Theta_0^M)  
    \geq \alpha^{-1})
     = 1.
\end{equation*}
\end{theorem}
\begin{proof}
For any $m$, it is clear that $\Gbar_M(\Theta_0^1,\ldots,\Theta_0^M) \geq mM^{-2} \, G_n^{(m)}(\Theta_0^m)$. Consequently, 
\[ \Pr( \Gbar_M(\Theta_0^1,\ldots,\Theta_0^M) \geq \alpha^{-1} ) \geq \Pr( \frac{m}{M^2} \, G_n^{(m)}(\Theta_0^m) \geq \alpha^{-1} ) = \Pr( G_n^{(m)}(\Theta_0^m) \geq \qty(\frac{m\alpha}{M^2})^{-1}). \]
If $m$ is such that $H_0^m$ is false in the sense above, and the stated assumptions hold, then Theorem 2 of \citet{dey2024} applies and the right-hand side above converges to 1 as $n \to \infty$, since $m\alpha M^{-2} \in (0,1)$ is fixed. 
\end{proof}

\Cref{thm:size,thm:power} state that (a)~if all null hypotheses are true, then we identify this fact with probability $1-\alpha$, and (b)~if at least one null hypothesis is false, then we identify this fact with arbitrarily high probability, for sufficiently large sample sizes. While \Cref{thm:power} does not require that \textit{any} of the GUe-values satisfy the strong central condition, we caution against using powerful tests that do not provide finite-sample validity guarantees.

\section{Simulations}
An interesting example where multiple testing in a risk minimization problem arises organically is in quantile regression.  In order to determine if some covariate $X$ has an impact on a response $Y$, it is insufficient to test just the mean or median response---consider, e.g., data generated as $Y\mid X \sim N(0, X^2)$ for scalar $X$---so one must generally test multiple quantiles $\tau_1, \ldots, \tau_M \in (0, 1)$ and determine if any of the $\tau_m$-specific regression coefficients are non-zero. Note that the $\tau$-quantile regression problem is induced by the loss function
\begin{equation*}
    \ell_\tau(\theta\altgiven X, Y) :=  (Y - X^\top \theta)\cdot \bigl\{ \tau - \mathbbm{1}(Y - X^\top \theta < 0) \bigr\}.
\end{equation*}
That is, the estimated coefficient vector $\thetahat^{(\tau_m)}$ for $\tau_m$-quantile regression is given by $\break\argmin_\theta n^{-1} \sum_{i=1}^n \ell_{\tau_m}(\theta\altgiven X_i, Y_i)$, and our true quantities of interest, $\theta^{(\tau_m)}$, are risk minimizers $\arg\min_\theta \E[\ell_{\tau_m}(\theta\altgiven X, Y)]$ for $m=1,\ldots,M$.  If the goal is to test for the significance of covariate, say, $X_1$, in the $\tau_m$-specific quantile regression, then that corresponds to a composite null hypothesis $H_0^m: \theta^{(\tau_m)} \in \Theta_0^m := \{\theta \in \mathbb{R}^p: \theta_1=0\}$; this is a composite null hypothesis because only the value of the $X_1$ coefficient is set to 0, the others (including the intercept) are free to vary.  For each $m$, we can construct a GUe-value $G_n^{(m)}(\Theta_0^m)$ to test $H_0^m$, but the relevant scientific question is not about a quantile-specific coefficient.  Instead, the question concerns {\em all} quantiles---or at least all of our $M$-many quantiles for large $M$---so this naturally turns into a multiple testing problem.  Our proposal is to combine the quantile-specific GUe-values into $\Gbar_M$ according to \eqref{eq:finalGue} and then reject the meta-hypothesis ``$X_1$ is not significant'' at level $\alpha \in (0,1)$ if and only if $\Gbar_M \geq \alpha^{-1}$. 


We now provide two simulations demonstrating the efficacy of the GUe-value in this quantile regression problem. We note that in the following examples, the strong central condition easily holds because the data are bounded.  We also use the suggested Algorithm 1 of \citet{dey2024} to select the learning rate for each GUe-value. 
\begin{figure}[t]
    \centering
    \begin{subfigure}{0.4\textwidth}
            \includegraphics[width=\textwidth]{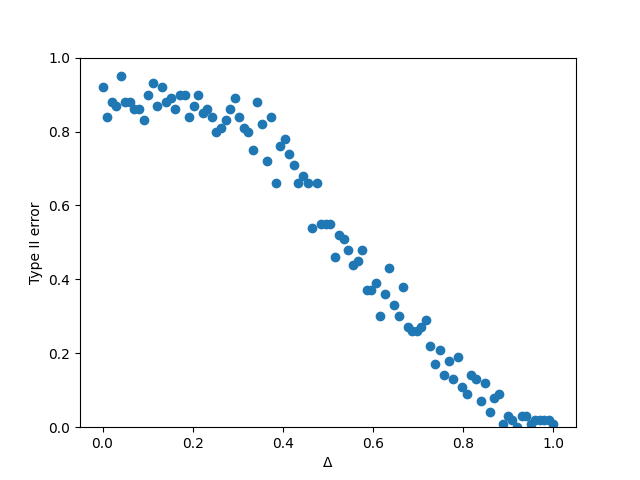}
            \caption{}\label{fig:triangle_delta_power}
    \end{subfigure}
    \begin{subfigure}{0.4\textwidth}
            \includegraphics[width=\textwidth]{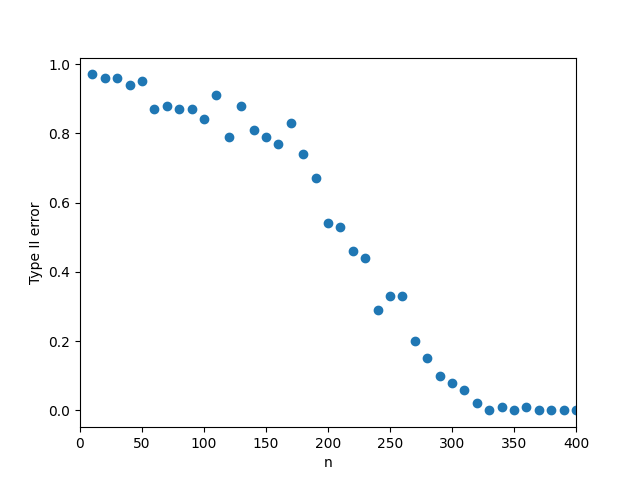}
            \caption{}\label{fig:triangle_n_power}
    \end{subfigure}
    \caption{Type II error rate of the combined GUe-value in \Cref{ex:triangle}, estimated using 100 Monte Carlo iterations. In (a), a sample size of 50 is held constant; in (b), $\Delta = 0.3$ is held constant.}\label{fig:trianglepower}
\end{figure}
\begin{example}\label{ex:triangle}
A simple example considers the following family of distributions on $[0, 1]^2$ parameterized by $\Delta\in[0, 1]$:
\begin{equation*}
    Y\given X \sim \operatorname{Uniform}\qty(\tfrac{X \Delta}{2}, 1 - \tfrac{X\Delta}{2}), \qquad X \sim \operatorname{Uniform}(0, 1)
\end{equation*}
At $\Delta = 0$, all of the $\tau$-specific regression coefficients are 0, whereas, if $\Delta > 0$, then only the $\tau=\frac12$ coefficient is 0.  Furthermore, as $\Delta$ increases, the ``signal" of a nonzero coefficient becomes easier to detect. For a given $\Delta$, we perform $\tau$-quantile regression on a sample of size $n=50$ for each $\tau=\tau_m$ in $\{0.02, 0.04, \ldots, 0.98\}$, obtain the $M=49$ many GUe-values for testing if $\theta_1^{(\tau_m)}=0$ on the sample, combine them to get $\Gbar_M$, and carry out the overall significance test at level $\alpha = 0.1$.  
The Type II error rate for different choices of $\Delta$ is plotted in \Cref{fig:triangle_delta_power}. We see that near $\Delta = 0$, the Type II error rate is approximately $1-\alpha$, as one would expect, but the error rate quickly diminishes as it gets ``easier" to identify non-zero slopes in the quantile regressions. We can similarly examine how the Type II error decreases as the sample size increases; this is shown in \Cref{fig:triangle_n_power}.
We finally note that at $\Delta = 0$, the empirical FDR for the GUe-values before being merged was 0.04---about half the allowed limit of $0.10$.
\end{example}

\begin{example}\label{ex:trapezoid}
Next, consider the following family of distributions parameterized by $\Gamma\in[0, 1]$:
\begin{equation*}
    Y\mid X\sim \operatorname{Uniform}(0, 1 + \Gamma\cdot (X-1)), \qquad X \sim \operatorname{Uniform}(0, 1)
\end{equation*}
As before, this family has the property that for $\Gamma \approx 0$, all signals are weak, whereas the signals for $\Gamma \approx 1$ are strong. Compared to \Cref{ex:triangle}, there are ``fewer" quantiles with large slope coefficients, but they tend to be larger in magnitude and, therefore, easier to detect. In terms of power as $\Gamma$ and $n$ increases in \Cref{fig:trapezoid_power}, here we find that the GUe test is generally more powerful than in \Cref{ex:triangle}, with the Type~II error rate dropping off more quickly as either $\Gamma$ or $n$ increases. This suggests that our testing procedure does better at detecting a few strong signals as opposed to multiple weaker signals. Finally, we note that the empirical FDR at $\Gamma = 0$ remains at $0.04$ in this case. 
\begin{figure}[t]
    \centering
    \begin{subfigure}{0.4\textwidth}
            \includegraphics[width=\textwidth]{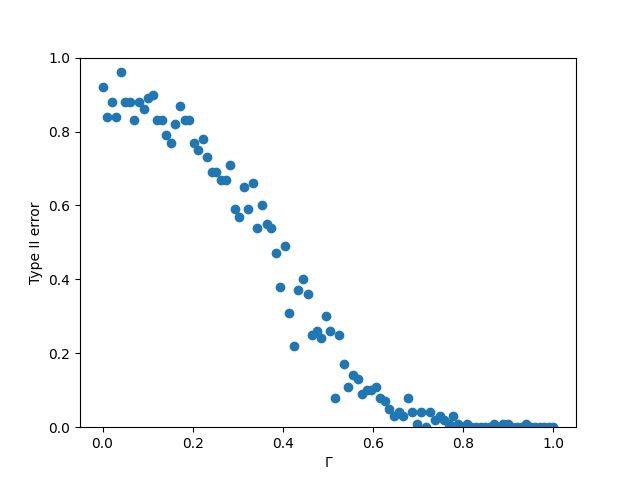}
            \caption{}\label{fig:trapezoid_delta_power}
    \end{subfigure}
    \begin{subfigure}{0.4\textwidth}
            \includegraphics[width=\textwidth]{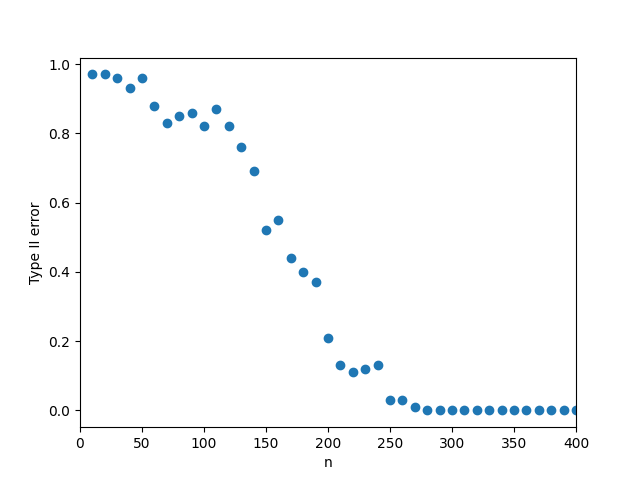}
            \caption{}\label{fig:trapezoid_n_power}
    \end{subfigure}
    \caption{Type II error rate of the combined GUe-value in \Cref{ex:trapezoid}, estimated using 100 Monte Carlo iterations. In (a), a sample size of 50 is held constant; in (b), $\Gamma = 0.3$ is held constant.}\label{fig:trapezoid_power}
\end{figure}
\end{example}

\begin{example}
It is also of interest to consider how the learning rate for the GUe-value is selected for various quantiles and data-generating distributions---independent of their use in hypothesis testing. Thus, we plot in \Cref{fig:lrs} the learning rates selected for the GUe-values for each $\tau$ for the data-generating distributions from \Cref{ex:triangle,ex:trapezoid}. The plots suggest that higher learning rates are chosen at more extreme quantiles; that this holds even when $\Delta$ and $\Gamma$ are zero---where there is no effect regardless of quantile---demonstrates that the choice of learning rate depends on the loss function $\ell_\tau$. However, the learning rate is not solely dependent on the loss function, but also on the ``ground truth" data-generating distribution. In particular, there is a stark increase in the learning rate as it becomes ``more obvious" that the null hypotheses should be rejected (i.e., as $\Delta$ and $\Gamma$ increase). This also is to be expected, as a larger learning rate tends to produce larger GUe-values which are more likely to reject the null hypothesis. We finally note that the learning rates chosen in \Cref{fig:trapezoid_lr} appear to be somewhat smaller than those chosen in \Cref{fig:triangle_lr} despite the hypothesis test having better power; this suggests that simply looking at the magnitude of the learning rate is not a particularly useful metric to analyze the Type II error rate when the data-generating family is not held constant. 
\begin{figure}[t]
    \centering
    \begin{subfigure}{0.4\textwidth}
        \includegraphics[width=\textwidth]{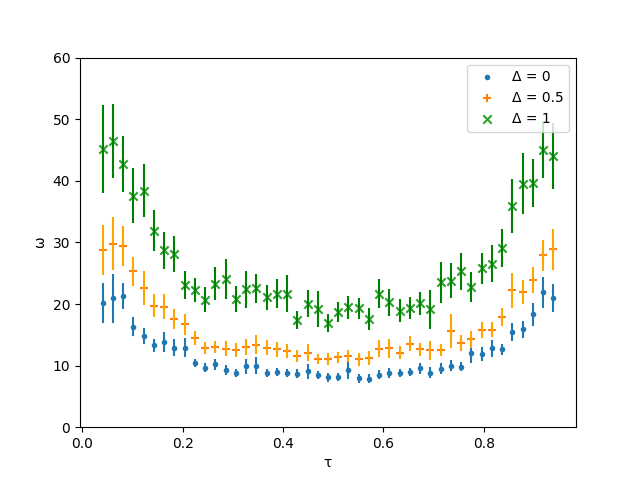}
        \caption{}\label{fig:triangle_lr}
    \end{subfigure}
    \begin{subfigure}{0.4\textwidth}
        \includegraphics[width=\textwidth]{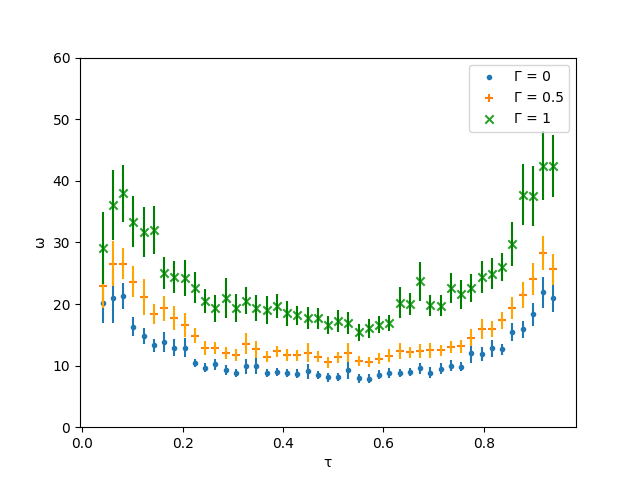}
        \caption{}\label{fig:trapezoid_lr}
    \end{subfigure}
    \caption{The average learning rate (over 100 Monte Carlo samples) chosen by Algorithm 1 of \citet{dey2024} at different $\tau$ values. Error bars are of width 1 standard error.}
    \label{fig:lrs}
\end{figure}
\end{example}

\section{Conclusion}
In this paper, we demonstrated that the GUe-value from the generalized universal inference framework is well-equipped to be used within the e-BH procedure to test multiple hypotheses regarding the locations of risk minimizers. We provided both theory and simulation results that show that the GUe-value retains the validity and power properties that it possesses for testing a single hypothesis. We also demonstrated the utility of multiple testing of risk minimizers in the specific problem of determining, in the absence of knowledge of the data-generating distribution, if there exists a linear association between two variables at any conditional quantile of the response.

There are a number of directions that still need to be explored. For example, \Cref{thm:power} is fairly weak; it is likely that more can be said regarding how the number of tests being carried out as well as the size of the signal impacts the growth rate (and hence the power) of the combined GUe-value. Additionally, the particular quantile regression testing problem we posed also leaves open several avenues for future work. For one, our simulations fixed fifty evenly spaced quantiles, and it would be interesting to explore how sensitive the GUe-value is to the locations and total number of tested quantiles. Furthermore, simply testing a finite number of quantiles does not fully address the problem: Although the strong dependence between GUe-values of neighboring quantiles means that the grid of tested quantiles likely does not need to be particularly dense, a full solution would require simultaneous testing of \textit{all} quantiles in the interval $[0, 1]$. A compromise between a fixed finite number of quantiles and an entire continuum of them would be to consider letting the number of quantiles grow with the sample size; in this setting, it may make sense to consider using the ``online" version of the GUe-value from \citet{dey2024}.  The exact properties of FDR-controlled and merged online GUe-values should be investigated in future work. 

\vskip 0.2in
\bibliography{references.bib}
\end{document}